\documentclass[12pt,a4paper]{elsarticle}

\usepackage{amsmath}
\usepackage{amssymb}
\usepackage{graphicx}
\usepackage{mdframed}
\usepackage{algorithm}
\usepackage{algorithmic}
\usepackage{makecell}

\usepackage{algorithm}
\usepackage{algorithmic}
\usepackage{fullpage}

\newcommand{\A}{\mathcal{A}}

\newcommand{\C}{\mathcal{C}}
\newcommand{\K}{\mathcal{K}}

\newcommand{\wone}{\mathsf{W[1]}}

\newcommand{\np}{\mathsf{NP}}
\newcommand{\conp}{\mathsf{coNP}}

\newcommand{\cupdot}{\mathbin{\mathaccent\cdot\cup}}

\newtheorem{theorem}{Theorem}

\newtheorem{lemma}{Lemma}

\newproof{proof}{Proof}

\begin{document}

\begin{frontmatter}
\title{Parameterized covering in semi-ladder-free hypergraphs}
\author{Sylvain Guillemot}
\ead{guillemo@free.fr}

\begin{abstract}
In this article, we study the parameterized complexity of the \textsc{Set Cover} problem restricted to semi-ladder-free hypergraphs, a class defined by Fabianski et al. [Proceedings of STACS 2019]  \footnote{Note that their results are formulated in terms of domination in bipartite graphs, while we find more convenient to use the language of hypergraphs.}. We observe that two algorithms introduced by Langerman and Morin [Discrete \& Computational Geometry 2005] in the context of geometric covering problems can be adapted to this setting, yielding simple FPT and kernelization algorithms for \textsc{Set Cover} in semi-ladder-free hypergraphs. We complement our algorithmic results with a compression lower bound for the problem, which proves the tightness of our kernelization under standard complexity-theoretic assumptions.
\end{abstract}
\end{frontmatter}

\section{\label{sec1}Introduction}

A \emph{hypergraph} is a pair $H = (V,E)$, where $V$ is its set of vertices and $E \subseteq 2^V$ is its set of edges. A \emph{cover} of $H$ is a set $F \subseteq E$ such that each vertex of $V$ belongs to at least one edge of $F$. The \textsc{Set Cover} problem takes a hypergraph $H$ and an integer $k$, and asks whether $H$ admits a cover of size at most $k$. This problem has been widely studied from the point of view of approximability and parameterized complexity, with an emphasis on structural restrictions that make it efficiently solvable. A well-known restriction is the fact of having bounded \emph{VC-dimension}, a notion introduced in \cite{VC71}. It is known that for hypergraphs of dual VC-dimension $d$, the \textsc{Set Cover} problem can be efficiently approximated \cite{BG95,ERS05}, in the sense that we have a polynomial-time algorithm that either concludes that there is no cover of size $k$, or finds a cover of size $O(d k \log k)$.

However, bounding the VC-dimension or its dual is not enough to make the \textsc{Set Cover} problem fixed-parameter-tractable in $k$. Indeed, it was proved in \cite{BKMN16} that the problem remains $\wone$-hard, even for hypergraphs having both VC-dimension and dual VC-dimension equal to 2. This motivates the search for subclasses of hypergraphs where the problem becomes fixed-parameter-tractable. We focus here on the class of \emph{$(d+1)$-semi-ladder free} hypergraphs; the precise definition is postponed to Section \ref{sec2}. This class is introduced in \cite{FPST19} \footnote{Our definition of $(d+1)$-semi-ladder-freeness corresponds to having semi-ladder-index at most $(d+1)$ in the sense of \cite{FPST19}.}, and it is shown here that \textsc{Set Cover} is solvable in $k^{O(d k)} ||H||$ time when $H$ is $(d+1)$-semi-ladder-free. This algorithm uses a novel algorithmic paradigm called progressive exploration, which is also applied to other problems such as finding distance $r$-dominating sets and distance $r$-independent sets in nowhere dense graphs.  

In this paper, we present new results on the parameterized complexity of \textsc{Set Cover} in $(d+1)$-semi-ladder-free hypergraphs. Our positive results are obtained by adapting two algorithms introduced in \cite{LM05} in the context of geometric covering problems. Our first algorithm uses a bounded-search-tree approach to solve \textsc{Set Cover} in $(d+1)$-semi-ladder-free hypergraphs in $O(k^{d k} k ||H||)$ time. We note that it is entirely different from the previously mentioned algorithm based on progressive exploration, although both algorithms achieve the same worst-case running time. Our second algorithm is a polynomial kernelization for the problem, which reduces an instance of \textsc{Set Cover} in $(d+1)$-semi-ladder-free hypergraphs to an equivalent instance of size $O(k^{d^2})$. In addition, we obtain a compression lower bound for the problem, which implies that the problem is unlikely to admit a kernel of size $O(k^{d^2-\epsilon})$ under standard complexity-theoretic assumptions.

We make a couple of observations here. First, the results of \cite{PRS12} providing a kernelization for \textsc{Dominating Set} in $K_{d,d}$-free graphs can be recovered as a consequence of our kernelization algorithm, though with slightly worse bounds. In particular, this encompasses the case of \textsc{Dominating Set} in $d$-degenerate graphs, for which a compression lower bound of $k^{(d-1)(d-3) - \epsilon}$ was already known \cite{CGH17}. Second, since semi-ladder-freeness is preserved by taking dual hypergraphs, it follows that our results also apply to the \textsc{Hitting Set} problem in $(d+1)$-semi-ladder-free hypergraphs. Third, it can be seen that $(d+1)$-semi-ladder-free hypergraphs have (dual) VC-dimension at most $d+1$, and thus the aforementioned approximation result also holds; it is unlikely that the $O(d k \log k)$ bound can be improved though, even for $d = 2$, as shown by \cite{KAR00}.

For additional background, we refer the reader to the textbooks \cite{B84} on hypergraphs, \cite{CFKLMPPS15} on parameterized algorithms and \cite{FLSZ19} on kernelization.

\section{\label{sec2}Preliminaries}

Let $H = (V,E)$ be a hypergraph and $d$ be a positive integer. A \emph{$d$-semi-ladder} in $H$ is a pair $(W,F)$ with (a) $W = \{w_0,\ldots,w_d\} \subseteq V$, $F = \{f_0,\ldots,f_d\} \subseteq E$, (b) for each $i \in [0;d]$, we have $w_i \notin f_i$, (c) for each $i,j \in [0;d]$, we have $i < j \Rightarrow w_i \in f_j$. A \emph{$d$-ladder} in $H$ is a pair $(W,F)$ with (a) $W = \{w_0,\ldots,w_d\} \subseteq V$, $F = \{f_0,\ldots,f_d\} \subseteq E$, (b) for each $i,j \in [0;d]$, we have $i < j \Leftrightarrow w_i \in f_j$. We say that $H$ is \emph{$d$-semi-ladder-free} (resp. \emph{$d$-ladder-free}) iff there is no $d$-semi-ladder (resp. $d$-ladder) in $H$.

Given $e,e' \in E$, we say that $e'$ \emph{covers} $e$ (in $H$) iff $e \subset e'$ and there is no edge $f \in E$ such that $e \subset f \subset e'$. Given $e \in E$, an \emph{$e$-chain} in $H$ is a chain of inclusions $e_0 \subset e_1 \subset \ldots \subset e_l$ with $e_0,\ldots,e_l \in E$ and $e_l = e$; the \emph{length} of the chain is $l$. We define $l_H(e)$ as the maximum length of an $e$-chain in $H$. We define $L(H)$ as the maximum of $l_H(e)$ for $e \in E$. 

We say that $H$ is \emph{intersection-closed} iff (a) $E$ contains the edge $V$, (b) for each $e,f \in E$, we have $e \cap f \in E$. In general, if $H$ is an arbitrary hypergraph, its \emph{intersection-closure} is the minimal hypergraph $H'$ such that (a) $H$ is a partial hypergraph of $H'$, (b) $H'$ is intersection-closed. Fix $S \subseteq V$. We let $K_H(S) = \{ e \in E : S \subseteq e \}$, and we define $M_H(S)$ equal: to $V$ (if $K_H(S) = \emptyset$), or to $\cap_{e \in K_H(S)} e$ (otherwise).

Theorem \ref{thm1} below gives a characterization of $(d+1)$-semi-ladder-free hypergraphs, in terms of their intersection-closure.

\begin{lemma} \label{lem1} Let $H = (V,E)$ be a hypergraph, let $e_1,e_2 \in E$, and let $H' = (V,E+\{e\})$ where $e = e_1 \cap e_2$. If $H'$ has a $d$-semi-ladder, then the same holds for $H$.
\end{lemma}

\begin{proof} Suppose that $H'$ has a $d$-semi-ladder $(W,F)$, with $W = \{w_0,\ldots,w_d\}$ and $F = \{f_0,\ldots,f_d\}$. If $e \notin F$, then $(W,F)$ is a $d$-semi-ladder in $H$. Suppose now that $e$ is equal to $f_i$, for some $i \in [0;d]$. It follows that $w_0,\ldots,w_{i-1}$ belong to both $e_1$ and $e_2$, and that $w_i$ does not belong to $e_j$ for some $j \in \{1,2\}$. Replacing $f_i$ with $e_j$ then yields a pair $(W,F')$ that is a $d$-semi-ladder in $H$. \qed
\end{proof}

\begin{theorem} \label{thm1} Let $H = (V,E)$ be a hypergraph, let $H'$ be its intersection-closure, and let $d$ be a positive integer. The following are equivalent:
\begin{itemize}
\item[(i)] $H$ is $(d+1)$-semi-ladder-free;
\item[(ii)] $H'$ is $(d+1)$-ladder-free;
\item[(iii)] $L(H') \leq d+1$.
\end{itemize}
\end{theorem}

\begin{proof} 
%
%
$(i) \Rightarrow (ii)$. Suppose that $H'$ contains a $(d+1)$-ladder. We may then apply Lemma \ref{lem1} to obtain a $(d+1)$-semi-ladder in $H$. 
 
$(ii) \Rightarrow (iii)$. Suppose that $L(H') > d+1$. It follows that $H'$ contains a chain $f_0 \subset \ldots \subset f_{d+1} \subset f_{d+2}$. For each $i \in [0;d+1]$, let us choose an element $w_i \in f_{i+1} - f_i$. Let $W = \{w_0,\ldots,w_{d+1}\}$ and $F = \{f_0,\ldots,f_{d+1}\}$. Then $(W,F)$ is a $(d+1)$-ladder in $H'$. 

$(iii) \Rightarrow (i)$. Suppose that $H$ contains a $(d+1)$-semi-ladder $(W,F)$, with $W = \{w_0,\ldots,w_{d+1}\}$ and $F = \{f_0,\ldots,f_{d+1}\}$. For each $i \in [0;d+1]$, let $S_i = \{ w_0,\ldots,w_{i-1} \}$, and let $e_i = M_{H'}(S_i)$. Since $H'$ is intersection-closed, it contains the edges $e_0,\ldots,e_{d+2}$. We claim that $e_i \subset e_{i+1}$ for each $i \in [0;d+1]$. On the one hand, since $S_i \subseteq f_i$ we have $e_i \subseteq f_i$ and thus $w_i \in e_{i+1} - e_i$. On the other hand, since $S_i \subseteq e_{i+1}$ we have $e_i \subseteq e_{i+1}$ by definition of $e_i$. We conclude that $e_0 \subset \ldots \subset e_{d+2}$ is a chain in $H'$, and thus $L(H') > d+1$, contradiction. \qed
\end{proof}

We say that the hypergraph $H$ is \emph{$d$-flat} iff it satisfies the conditions in the above theorem. We point out that $d$-flat hypergraphs are closed under taking dual hypergraphs, and under taking partial subhypergraphs.

In order to solve \textsc{Set Cover}, we would like to make a simplifying assumption. We say that $H$ is \emph{reduced} iff $\cap_{e \in E} e = \emptyset$. If $H$ is not reduced, we construct a hypergraph $H' = \textsc{Reduce}(H)$ as follows: we let $X = \cap_{e \in E} e$, and we let $H' = (V-X,\{ e - X : e \in E \})$. By construction, $H'$ is reduced. The following lemma ensures that to solve \textsc{Set Cover}, we may consider $H'$ instead of $H$.

\begin{lemma} \label{lem2} Let $H = (V,E)$ be a hypergraph and let $H' = \textsc{Reduce}(H)$. Then: $H$ has a cover of size at most $k$ iff $H'$ has a cover of size at most $k$.
\end{lemma}

\begin{proof} Let $X = \cap_{e \in E} e$ as above. Clearly, if $C = \{ e_1,\ldots,e_l \}$ is a cover of $H$, then $C' = \{ e_1 - X, \ldots, e_l - X \}$ is a cover of $H'$. Suppose that $C = \{ e'_1, \ldots, e'_l \}$ is a cover of $H'$, with $e'_1,\ldots,e'_l \subseteq V-X$. For each $i \in [l]$, let $e_i = e'_i \cup X$. By definition of $X$, we have $e_1,\ldots,e_l \in E$. We show that $C' = \{ e_1,\ldots,e_l \}$ is a cover of $H$. Indeed, an element of $X$ is covered by every edge $e_i$, and for $v \in V-X$ we have some $e'_i \in C'$ covering $v$, which implies that $e_i$ also covers $v$. \qed
\end{proof}

In the following, we will denote by \textsc{$d$-Flat Set Cover} the restriction of \textsc{Set Cover} to reduced $d$-flat hypergraphs.

The \emph{size} of $H$ is $||H|| = \sum_{e \in E} |e|$. We suppose that the hypergraph $H$ given as input is represented by adjacency lists; such a representation has size $||H||$. Note that the binary representation actually has a bitsize $O(||H|| \log |V|)$, but we will conveniently ignore the logarithmic factors. Theorem \ref{thm2} below bounds the size of a reduced $d$-flat hypergraph in terms of its number of vertices.

\begin{lemma} \label{lem3} Let $H = (V,E)$ be an intersection-closed hypergraph, let $e \in E$ and let $e_0 \subset e_1 \subset \ldots \subset e_l$ be a maximal $e$-chain in $H$. For each $i \in [0;l-1]$, let $v_i$ be an arbitrary element of $e_{i+1} - e_i$. Then: $M_H(\{v_0,\ldots,v_{l-1}\}) = e$.
\end{lemma}

\begin{proof} For each $i \in [0;l]$, let $S_i = \{v_0,\ldots,v_{i-1}\}$. We show by induction on $i$ that $M_H(S_i) = e_i$. This holds for $i = 0$ since $e_0$ is the unique minimal element of $H$ and thus $M_H(\emptyset) = e_0$. Suppose that the property holds for $i$ and let us prove it for $i+1$. We then have $e_i \subset e_{i+1}$, $M_H(S_i) = e_i$ and $v_i \in e_{i+1} - e_i$. Since $S_i \subseteq e_i$ and $v_i \in e_{i+1}$, we have $S_{i+1} \subseteq e_{i+1}$. It follows that $e_i \subset M_H(S_{i+1}) \subseteq e_{i+1}$. Since $e_{i+1}$ covers $e_i$ in $H$, we conclude that $M_H(S_{i+1}) = e_{i+1}$. \qed
\end{proof}

\begin{theorem} \label{thm2} Let $H = (V,E)$ be a reduced $d$-flat hypergraph. Let $n = |V|$. Then: $||H|| = O(n^d)$.
\end{theorem}

\begin{proof} It suffices to prove the property when $H$ is intersection-closed. Let $E' = E - \{ V \}$. Let $P$ denote the set of pairs $(e,v)$ with $e \in E, v \in e$, and let $P'$ denote the set of pairs $(e,v)$ with $e \in E', v \in e$. Let $Q$ denote the set of pairs $(S,v)$ with $S \in [V]^{\leq d}$ and $v \in S$. We define the mapping $F : Q \rightarrow P$ such that: given $q = (S,v) \in Q$, we let $F(q) = (M_H(S),v)$. Observe that we have $F(q) \in P$: we have $M_H(S) \in E$ by definition, and since $v \in S$ and $S \subseteq M_H(S)$ we have $v \in M_H(S)$.

We show that every element of $P'$ is in the range of $F$. Consider an element $p = (e,v) \in P'$, with $e \in E'$ and $v \in e$. Let $e_0 \subset e_1 \subset \ldots \subset e_l$ be a maximal $e$-chain in $H$; we have $l \leq d$ by Theorem \ref{thm1}. Since $H$ is reduced, we have $e_0 = \emptyset$, and thus there is an index $i \in [0;p-1]$ such that $v \in e_{i+1} - e_i$. We may apply Lemma \ref{lem3} with $v_i = v$ to obtain a set $S = \{v_0,\ldots,v_{l-1}\}$ containing $v$ such that $M_H(S) = e$. Thus, if we let $q = (S,v)$ we have $q \in Q$ and $F(q) = p$.

We deduce from the above that $|P'| \leq |Q|$. Let $Z = \sum_{e \in E'} |e|$. Observe that $Z = |P'|$ and $|Q| \leq n^d$, hence we have $Z \leq n^d$. We conclude by observing that $||H|| = Z + n = O(n^d)$. \qed
\end{proof}

%
%
%

We observe that a similar reasoning as in Theorem \ref{thm2} shows that for $H = (V,E)$ $d$-flat hypergraph, we have $|E| \leq 1 + \sum_{i = 0}^{d} \binom{n}{i}$. We note that the Sauer-Shelah Lemma \cite{Sa72} provides a weaker bound $|E| \leq \sum_{i = 0}^{d+1} \binom{n}{i}$, using the fact that $d$-flat hypergraphs have VC-dimension at most $d+1$.


\section{\label{sec3}Algorithmic results}

\subsection{FPT algorithm}

In this section, we describe an FPT algorithm for \textsc{$d$-Flat Set Cover}. It is an adaptation of the algorithm \textsc{Bst-Dim-Set-Cover} of \cite{LM05}. Its pseudocode is given in Algorithm \ref{algo1} below, and its correctness is stated in Theorem \ref{thm4}.

Suppose that we are given a hypergraph $H = (V,E)$ and an integer $k$. Let $H' = (V,E')$ be the intersection-closure of $H$. We will need the following definitions. A \emph{$k$-edge-tuple in $H$} is a tuple $t = (e_1,\ldots,e_k)$ where $e_1,\ldots,e_k$ are in $E'$ and distinct from $V$. Given $i \in [k]$ and $e \in E'$ distinct from $V$, we denote by $t[i \rightarrow e]$ the $k$-edge-tuple $t'$ such that $t'[i] = e$ and $t'[j] = t[j]$ ($j \neq i$). A \emph{solution} for $(H,k,t)$ is a set $S = \{f_1,\ldots,f_k\}$ such that (a) $S \subseteq E$, (b) for each $i \in [k]$, $e_i \subseteq f_i$, (c) $S$ is a cover of $H$.

The following algorithm $\textsc{SolveSetCover}(H,k,t)$ takes a hypergraph $H$, an integer $k$ and $t$ $k$-edge-tuple in $H$, and decides the existence of a solution for $(H,k,t)$.

\begin{algorithm}[!]
\caption{$\textsc{SolveSetCover}(H,k,t)$}
\label{algo1}
\begin{algorithmic}[1]
\STATE suppose that $t = (e_1,\ldots,e_k)$
\STATE let $C = e_1 \cup \ldots \cup e_k$
\IF{$C = V$}
\STATE for each $i \in [k]$, let $f_i$ be an edge of $E$ containing $e_i$
\STATE return $\{f_1,\ldots,f_k\}$
\ENDIF
\STATE choose $v \in V - C$
\FOR{$i$ from $1$ to $k$}
\STATE let $f = M_H(e_i+v)$
\IF{$f \neq V$}
\STATE $r \leftarrow \textsc{SolveSetCover}(H,k,t [i \rightarrow f])$
\STATE if $r \neq \perp$ then return $r$
\ENDIF
\ENDFOR
\STATE return $\perp$
\end{algorithmic}
\end{algorithm}

\begin{lemma} \label{lem4} Consider a call $\textsc{SolveSetCover}(H,k,t)$, where $H = (V,E)$, $V \notin E$, and $t$ is a $k$-edge-tuple in $H$. If the call returns $S \neq \perp$, then $S$ is a solution for $(H,k,t)$; if the call returns $\perp$, then there is no solution for $(H,k,t)$.
\end{lemma}

\begin{proof} We reason by induction on the height of the call tree. 

Suppose first that the current call exits in Line 5. For each $i \in [k]$, since $e_i \in E'$ and $e_i \neq V$, Line 4 can find an edge $f_i$ as stated. Since $C = V$, it follows that $\{ f_1,\ldots,f_k \}$ is a solution for $(H,k,t)$. Suppose now that the current call exits in Line 12. Thus, the call in Line 11 has returned $S \neq \perp$. Consider the tuple $t' = t[i \rightarrow f]$. Since $f \neq V$, we infer that $t'$ is a $k$-edge-tuple in $H$. By induction hypothesis, $S$ is a solution for $(H,k,t')$. Since $e_i \subset f$, $S$ is also a solution for $(H,k,t)$. 

Suppose finally that the current call exits in Line 15. Let us suppose by contradiction that there exists $S = \{f_1,\ldots,f_k\}$ solution for $(H,k,t)$. Let $v \in V-C$ as chosen in Line 7. By definition of $S$, we have $v \in f_i$ for some $i \in [k]$. Consider the $i$th step of the loop in Lines 8-14, and consider $f$ as defined in Line 9, and $t' = t[i \rightarrow f]$. Since $S$ is a solution for $(H,k,t)$ and $v \in f_i$, we have $e_i+v \subseteq f_i$. It follows that $f_i \in K_H(e_i+v)$, and thus $f \neq V$. By definition, we have $f \subseteq f_i$, and thus $S$ is a solution for $(H,k,t')$. But then the recursive call in Line 11 would not have returned $\perp$, contradiction. \qed
\end{proof}

\begin{theorem} \label{thm4} Fix an instance $(H,k)$ of \textsc{$d$-Flat Set Cover} with $H = (V,E)$ and $V \notin E$. Let $e_0 = \emptyset$, and let $t_0 = (e_0,\ldots,e_0)$. The call $\textsc{SolveSetCover}(H,k,t_0)$ runs in $O(k^{d k} k ||H||)$ time, and correctly solves the instance $(H,k)$.
\end{theorem}

\begin{proof} The correctness follows from Lemma \ref{lem4}. Let us justify the running time.

We first consider the time taken by a recursive call, excluding subcalls. First, observe that the instructions in Lines 2-6 take $O(k ||H||)$ time. Now, for a fixed $i \in [k]$, Line 9 takes $O(||H||)$ time, hence the loop in Lines 8-14 takes $O(k ||H||)$ time. Thus, a recursive call takes $O(k ||H||)$ time in total.

We now consider the number of recursive calls. Let $H'$ be the intersection-closure of $H$. Given $t$ $k$-edge-tuple in $H$, we define its measure as $m(t) = d k - \sum_{i = 1}^{k} l_{H'}(e_i)$. Since $H$ is $d$-flat, it follows from Theorem \ref{thm1} that for $t$ $k$-edge-tuple in $H$, we always have $m(t) \geq 0$. Moreover, we have: (a) $t_0$ is a $k$-edge-tuple in $H$ such that $m(t_0) = dk$; (b) for a call $\textsc{SolveSetCover}(H,k,t)$ with $t$ $k$-edge-tuple in $H$, for each recursive call $\textsc{SolveSetCover}(H,k,t')$ we have $t'$ $k$-edge-tuple in $H$ and $m(t') \leq m(t)-1$. Since each call issues $k$ recursive calls, the total number of recursive calls is thus $O(k^{d k})$.

From the above points, we conclude that the initial call to $\textsc{SolveSetCover}(H,k,t_0)$ takes $O(k^{d k} k ||H||)$ time. \qed
\end{proof}

\subsection{Kernelization algorithm}

We now describe a kernelization for \textsc{$d$-Flat Set Cover}. It is an adaptation of the algorithm \textsc{Kernelize} of \cite{LM05}. The kernelization is described in Algorithm \ref{algo2} below, and its correctness is stated in Theorem \ref{thm5}.


Consider a hypergraph $H = (V,E)$ and a set $S \subseteq V$. The operation of \textit{grouping $S$} produces the hypergraph $H' = (V',E')$, where $V' = V - S + \{v\}$, $v$ is a new vertex not in $V$, and $E'$ contains (a) for each $e \in E$ such that $S \nsubseteq e$, the edge $e - S$, (b) for each $e \in E$ such that $S \subseteq e$, the edge $e - S + \{v\}$. We denote by $\textsc{Group}(H,S)$ the result of this operation.

The following algorithm $\textsc{KernelizeSetCover}(H,k)$ takes a hypergraph $H$ and an integer $k$, and computes a kernel for the instance $I = (H,k)$.

\begin{algorithm}[!]
\caption{$\textsc{KernelizeSetCover}(H,k)$}
\label{algo2}
\begin{algorithmic}[1]
\STATE $cont \leftarrow true$
\WHILE{$cont$}
\STATE let $H' = (V',E')$ be the intersection-closure of $H$
\STATE let $Z = \{ (e,i) : e \in E', i > 0, l_{H'}(e) = i \text{ and } |e| > k^{i-1} \}$
\IF{$Z = \emptyset$}
\STATE $cont \leftarrow false$
\ELSE
\STATE choose $(e,i) \in Z$ with $i$ minimum
\STATE $H \leftarrow \textsc{Group}(H,e)$
\ENDIF
\ENDWHILE
\STATE return $H$
\end{algorithmic}
\end{algorithm}

\begin{lemma} \label{lem5} Consider a hypergraph $H = (V,E)$, a set $S \subseteq V$, and let $H' = \textsc{Group}(H,S)$. If $H$ is $d$-flat, then so is $H'$.
\end{lemma}

\begin{proof} Suppose that $H' = (V',E')$, with $V' = V - S + \{v\}$. Suppose by contradiction that $H'$ contains a $(d+1)$-semi-ladder $(W,F)$ with $W = \{w_0,\ldots,w_{d+1}\}$ and $F = \{f_0,\ldots,f_{d+1}\}$. For each $i \in [0;d+1]$, let $f'_i$ be the edge of $H$ corresponding to $f_i$, and let $F' = \{f'_0,\ldots,f'_{d+1}\}$. If $v \notin W$, then $(W,F')$ is a $(d+1)$-semi-ladder in $H$. Suppose now that $v \in W$, then $v = w_i$ for some $i \in [0;d+1]$. Since $w_i \notin f_i$, we have $S \nsubseteq f'_i$, and thus we find an element $w' \in S$ such that $w' \notin f'_i$. For each $j > i$, since $w_i \in f_j$ we have $S \subseteq f'_j$ and thus $w \in f'_j$. Thus, if we replace $w_i$ with $w'$ we obtain $(W',F')$ $(d+1)$-semi-ladder in $H$. \qed
\end{proof}

\begin{lemma} \label{lem6} Consider the instruction in Line 9, and let $H_1,H_2$ be the old and new values of $H$. The instances $(H_1,k)$ and $(H_2,k)$ are equivalent.
\end{lemma}

\begin{proof} Suppose that $H_i = (V_i,E_i)$ for $i \in \{1,2\}$. We then have $V_2 = V_1 - e + \{v\}$. We show that the instances $(H_1,k)$ and $(H_2,k)$ are equivalent.

Suppose that $(H_2,k)$ is a yes-instance. Let $S = \{f_1,\ldots,f_l\}$ be a cover of $H_2$ of size at most $k$. For each $i \in [l]$, suppose that the edge $f_i$ of $H_2$ comes from the edge $e_i$ of $H_1$. We show that $S' = \{e_1,\ldots,e_l\}$ is a cover of $H_1$. For a vertex $u \in V_1 - e$, we have $u \in f_i$ for some $i \in [l]$, and thus $u \in e_i$. For a vertex $u \in e$, we have $v \in f_i$ for some $i \in [l]$, and thus $e \subseteq e_i$ and $u \in e_i$. 

Suppose that $(H_1,k)$ is a yes-instance. Let $S = \{e_1,\ldots,e_l\}$ be a cover of $H_1$ of size at most $k$. For each $i \in [l]$, let $f_i$ be the edge of $H_2$ coming from $e_i$. We show that $S' = \{f_1,\ldots,f_l\}$ is a cover of $H_2$. For a vertex $u \in V_1 - e$, we have $u \in e_i$ for some $i \in [l]$, and thus $u \in f_i$. It remains to show that $v$ is covered by $S'$. This is clear if we have $e \subseteq e_i$ for some $i \in [l]$. Suppose that this is not the case. For each $i \in [l]$, let $e'_i = e_i \cap e$; then $e'_i$ belongs to $H'$, $e'_i \subset e$ and thus $l_{H'}(e'_i) < i$. By choice of $(e,i)$, we have $|e'_i| \leq k^{i-2}$. But then $\sum_{i = 1}^{l} |e'_i| \leq k^{i-1}$ and thus there is a vertex in $e$ not covered by $S$, contradiction. \qed
\end{proof}

\begin{theorem} \label{thm5} Fix an instance $(H,k)$ of \textsc{$d$-Flat Set Cover}, where $H$ has $n$ vertices. The call $\textsc{KernelizeSetCover}(H,k)$ computes in $O(n^{2d+2})$ time an equivalent instance $(H_r,k)$ of size $||H_r|| = O(k^{d^2})$.
\end{theorem}

\begin{proof} Consider the execution of \textsc{KernelizeSetCover} on $(H,k)$. Suppose that the while loop executes $r$ steps, and for $s \in [0;r]$ let $H_s = (V_s,E_s)$ be the value of $H$ at the end of step $s$. We then have $H_0 = H$, and the algorithm returns $H_r$. A straightforward induction based on Lemmas \ref{lem5} and \ref{lem6} shows that at each step, we have $H_s$ $d$-flat, reduced and $(H_s,k)$ equivalent to $(H,k)$. Applying the induction hypothesis at the last step proves that $H_r$ is $d$-flat, reduced and $(H_r,k)$ is equivalent to $(H,k)$.

We now justify the kernel bound. Consider the last step of the while loop, we then have $H = H_r$. Let $H',Z$ as defined in Lines 3-4 in this step, with $H' = (V_r,E'_r)$. By Theorem \ref{thm1}, as $H_r$ is $d$-flat we have $L(H') \leq d+1$ and thus $H'$ contains the edge $V_r$ with $l_{H'}(V_r) \leq d+1$. Since $Z = \emptyset$, it follows that $|V_r| \leq k^d$. By Theorem \ref{thm2}, as $H_r$ is reduced and $d$-flat we have $|| H_r || = O(k^{d^2})$.

We finally justify the running time. We claim that a step $s$ of the while loop takes $O(n^{2d+1})$ time. Let $H'_s$ as defined in Line 3 in this step, with $H'_s = (V_s,E'_s)$. Since $H_s$ is $d$-flat and reduced, we have $H'_s$ $d$-flat and reduced by Theorem \ref{thm1}. Since $|V_s| \leq n$, it follows from Theorem \ref{thm2} that $|| H'_s || = O(n^d)$. Thus, $H'_s$ can be constructed in $O(n^{d+1})$ time. Next, computing the values $l_{H'_s}(e)$ for $e \in E'_s$ takes $O(n^{2d+1})$ time. All other operations can be performed in $O(n^{d+1})$ time. Since $r \leq n$, the algorithm takes $O(n^{2d+2})$ time in total. \qed
\end{proof}


\section{\label{sec4}Compression lower bound}

In this section, we show that \textsc{$d$-Flat Set Cover} is unlikely to have a kernel of size $O(k^{d^2-\epsilon})$. This is stated in Theorem \ref{thm7} below. We note that our result is phrased in terms of compressions, following \cite{BJK14}.

We first define an auxiliary problem called \textsc{$d$-Constraint-Cover}, where $d$ is a positive integer. Let $X$ be a set. We define inductively a notion of \emph{$i$-constraint over $X$}, for $i$ positive integer. 
\begin{itemize}
\item A \emph{$0$-constraint} over $X$ is an equality $C = (x = a)$ where $x \in X$ and $a \in \{0,1\}$. We define $V(C) = \{x\}$.
\item For $i > 0$, a \emph{$i$-constraint} over $X$ is a set $C$ of $(i-1)$-constraints over $X$, such that for $x,x' \in C$ distinct we have $V(x) \cap V(x') = \emptyset$. We define $V(C) = \cup_{x \in C} V(x)$.
\end{itemize}
The \emph{full set} of 0-constraints over $X$ is the set $\C^0 = \{ (x = a) : x \in X, a \in \{0,1\} \}$. A set $S \subseteq \C^0$ is an \emph{assignment} iff for each $x \in X$, there is a unique $a \in \{0,1\}$ such that $S$ contains $(x = a)$.

The \textsc{$d$-Constraint-Cover} problem is defined as follows. An instance $I$ consists of: (a) a parameter $k$, (b) a set of variables $X$ of size $d^2 k$, (c) the full set of $0$-constraints over $X$ denoted by $\C^0$, (d) a set of 1-constraints over $X$ of the form $\C^1 \subseteq [\C^0]^{\leq d}$, (e) a set of 2-constraints over $X$ of the form $\C^2 \subseteq [\C^1]^d$. A \emph{solution} for $I$ consists of: (a) an assignment $\K^0 \subseteq \C^0$ (of size $d^2 k$), (b) a set $\K^1 \subseteq \C^1$ of size $d k$, (c) a set $\K^2 \subseteq \C^2$ of size $k$, such that:
\begin{itemize}
\item[(i)] we have $\bigcup_{C \in \K^1} C = \K^0$;
\item[(ii)] we have $\bigcup_{C \in \K^2} C = \K^1$;
\item[(iii)] for $C \in \C^1$, we have $C \subseteq \K^0$ implies $C \in \K^1$.
\end{itemize}

We say that a 1-constraint $C \in \C^1$ is \emph{satisfied} by $\K^0$ iff $C \subseteq \K^0$; Point (iii) expresses that the elements of $\K^1$ are the only constraints of $\C^1$ satisfied by $\K^0$.
We observe that since $|\K^0| = d^2 k$ and $|\K^1| = d k$, it follows from Point (i) that the constraints in $\K^1$ have arity $d$ and that they form a partition of $\K^0$. In particular, this implies by Point (iii) that the constraints of $\C^1$ having arity less than $d$ cannot be satisfied in a solution for $I$. Similarly, since $|\K^1| = d k$ and $|\K^2| = k$, it follows from Point (ii) that the constraints in $\K^2$ form a partition of $\K^1$.

To obtain our compression lower bound for \textsc{$d$-Flat Set Cover}, we will give a reduction from \textsc{$d$-Constraint-Cover} leaning on the following result proved in the appendix.

\begin{theorem} \label{thm6} For every integer $d \geq 3$, \textsc{$d$-Constraint-Cover} has no compression of size $O(k^{d^2 - \epsilon})$, unless $\conp \subseteq \np / poly$.
\end{theorem}

We reduce \textsc{$d$-Constraint-Cover} to \textsc{$d$-Flat Set Cover}. Consider an instance $I$ of \textsc{$d$-Constraint-Cover}, consisting of the parameter $k$, the set $X$ of size $d^2 k$, the set $\C^0$ of 0-constraints over $X$, the set $\C^1 \subseteq [\C^0]^{\leq d}$ of $1$-constraints over $X$, and the set $\C^2 \subseteq [\C^1]^d$ of $2$-constraints over $X$.

We construct the following instance $I' = (H',k')$ of \textsc{$d$-Flat Set Cover}. We have $H' = (V',E')$. The set $V'$ consists of: (a) a vertex $v^1_x$ ($x \in X$), (b) a vertex $v^2_{x,a}$ ($x \in X, a \in \{0,1\}$), (c) a vertex $v^3_C$ ($C \in \C^1$). The set $E'$ consists of:
\begin{itemize}
\item[(a)] for $x \in X, a \in \{0,1\}$, an edge $e^1_{x,a} = \{ v^1_x, v^2_{x,a} \} \cup \{ v^3_C : C \in \C^1 \text{ containing } (x = a) \}$;
\item[(b)] for $C \in \C^1$, an edge $e^2_C = \{ v^2_{x,a} : C \text{ contains } (x = a) \}$;
\item[(c)] for $C \in \C^2$, an edge $e^3_C = \{ v^3_{C'} : C' \in \C^1 \text{ s.t. $C$ contains $C'$ } \}$.
\end{itemize}
The parameter is $k' = (d^2 + d + 1) k$.

A \emph{$d$-square} in $H'$ is a pair $(W,F)$ where $W \subseteq V'$, $F \subseteq E'$, $W = \{w_1,\ldots,w_d\}$, $F = \{f_1,\ldots,f_d\}$, and: for each $(i,j) \in [d] \times [d]$ such that $i \leq 2$ or $j \leq 2$, we have $w_i \in f_j$. We say that $H'$ is \emph{$d$-square-free} if it contains no $d$-square.

\begin{lemma} \label{lem7} $H'$ is $d$-square-free.
\end{lemma}

\begin{proof} Suppose by contradiction that $H'$ contains a $d$-square $(W,F)$ with $W = \{w_1,\ldots,w_d\}$ and $F = \{f_1,\ldots,f_d\}$. Let $W_1$ be the set of vertices $w_i$ of the form $v^1_x$, let $W_2$ be the set of vertices $w_i$ of the form $v^2_{x,a}$, and let $W_3$ be the set of vertices $w_i$ of the form $v^3_C$. We have the following cases.

Case 1: $W_1 \neq \emptyset$. Suppose that $W_1$ contains vertex $w_i$. Since $w_i$ belongs to $f_1,f_2$, we have $f_1 = e^1_{x,a}$ and $f_2 = e^1_{x,1-a}$. But then $f_1 \cap f_2 = \{v^1_x\}$, contradiction.

Case 2: $W_1 = \emptyset$ and $w_1,w_2 \in W_2$. We then have $x_1,x_2 \in X$ and $a_1,a_2 \in \{0,1\}$ such that $w_i = v^2_{x_i,a_i}$. Since $f_1,f_2$ contain $w_1,w_2$, we have $C,C' \in \C^1$ such that $f_1 = e^2_C, f_2 = e^2_{C'}$. Then all vertices $w_i$ are in $W_2$, and for each $i \in [d]$ we have $x_i \in X$ and $a_i \in \{0,1\}$ such that $w_i = v^2_{x_i,a_i}$. But then we have $C = C' = \{ x_1 = a_1, \ldots, x_d = a_d \}$, contradiction.

Case 3: $W_1 = \emptyset$ and $w_1 \in W_3$. Suppose first that some vertex $w_i$ is in $W_2$. We then have $x \in X$ and $a \in \{0,1\}$ such that $w_i = v^2_{x,a}$. Since $e_1,e_2$ are incident to $w_1$ and $w_i$, we have $e_1 = e_2 = e^1_{x,a}$, contradiction. Suppose now that each vertex $w_i$ is in $W_3$. For each $i \in [d]$, we then have $C_i \in \C^1$ such that $w_i = v^3_{C_i}$. Let $F_1$ be the set of edges $f_i$ of the form $e^1_{x,a}$, and let $F_3$ be the set of edges $f_i$ of the form $e^3_C$. We have the following subcases.

Case 3.1: we have $F_1,F_3 \neq \emptyset$. Consider $f_i \in F_1$ and $f_j \in F_3$. We then have $f_i = e^1_{x,a}$ for $x \in X, a \in \{0,1\}$, and $f_j = e^3_C$ for $C \in \C^2$. Since $f_j$ contains $w_1,w_2$, it follows that $C_1,C_2 \in C$, and thus $V(C_1) \cap V(C_2) = \emptyset$. Since $f_i$ contains $w_1,w_2$, it follows that $C_1,C_2$ both contains $(x = a)$, contradiction.

Case 3.2: we have $F_1 = \emptyset$. We then have $f_1 = e^3_C$ and $f_2 = e^3_{C'}$, for $C,C' \in \C^2$. Since $C,C'$ are in $\C^2$, we must have $C = C' = \{ C_1,\ldots,C_d \}$, contradiction.

Case 3.3: we have $F_3 = \emptyset$. For each $i \in [d]$, we have $f_i = e^1_{x_i,a_i}$ for $x_i \in X,a_i \in \{0,1\}$. It follows that for each $i \in [d]$, $C_1$ and $C_2$ both contain $(x_i = a_i)$. Then, we must have $C_1 = C_2 = \{ x_1 = a_1, \ldots, x_d = a_d \}$, contradiction. \qed
\end{proof}

\begin{lemma} \label{lem8} Let $S$ be a cover of $H'$ of size at most $k'$. We then have $\K^0 \subseteq \C^0$ assignment, $\K^1 \subseteq \C^1$ of size $dk$, $\K^2 \subseteq \C^2$ of size $k$, and a partition of $S$ into $S_1,S_2,S_3$ such that: $S_1 = \{ e^1_{x,a} : (x = a) \notin \K^0 \}, S_2 = \{ e^2_C : C \in \K^1 \}$ and $S_3 = \{ e^3_C : C \in \K^2 \}$.
\end{lemma}

\begin{proof} Among the edges of $S$, we let $S_1$ denote the edges of the form $e^1_{x,a}$, $S_2$ denote the edges of the form $e^2_C$, and $S_3$ denote the edges of the form $e^3_C$.

For $s \in \{1,2\}$, we let $X_s$ denote the set of elements $x \in X$ such that $|S \cap \{e^1_{x,0},e^1_{x,1}\} | = s$. Since $v^1_x$ can only be covered by an edge $e^1_{x,a}$, it follows that $X = X_1 \cupdot X_2$. Let $p = |X_2|$. Since $d^2 k = |X| = |X_1| + |X_2|$, it follows that $|S_1| = d^2 k + p$ and $|X_1| = d^2 k - p$. Thus, we have $|S_2| + |S_3| \leq (d + 1) k - p$ and $p \leq (d+1) k$. Let $Z$ be the set of vertices $v^2_{x,a}$ ($x \in X, a \in \{0,1\}$) such that $e^1_{x,a} \notin S$, and let $Z'$ be the set of vertices $v^2_{x,a}$ ($x \in X, a \in \{0,1\}$) such that $e^1_{x,a} \in S$. We have $|Z| = |X_1| = d^2 k - p$. Consider an integer $i$ such that $i (d-1) \leq p < (i+1) (d-1)$. Since the edges of $S_2$ cover $Z$, we must have $|S_2| \geq d k - i$. We then have an integer $s$ such that $|S_2| = d k + s$, with $-i \leq s \leq k - p$. Let $S'_2$ be the set of edges $f \in S_2$ such that $f \cap Z' = \emptyset$, and let $q = |S'_2|$. Observe that each edge $f \in S_2 - S'_2$ covers at most $d-1$ vertices of $Z$.

Let $v = q + (d - 1) (d k + s)$ and let $v' = d^2 k - p$. Since the edges of $S_2$ cover $Z$, we must have $v - v' \geq 0$. It follows that:
\begin{align*}
q &\geq d^2 k - p - (d - 1) (d k + s)\\
&= d (k - s) + (s - p)\\
&> d (k - s) + (-i - (i+1) (d-1))\\
&= d (k - s - (i+1)) + 1
\end{align*}
Now, for each edge $e^2_C \in S'_2$, we must have $v^3_C$ covered by an edge $e^3_{C'} \in S_3$. If we had $|S_3| \leq k - s - (i+1)$, we would obtain $q \leq d (k - s - (i+1))$. Thus, we have $|S_3| \geq k - s - i$, and it follows that: $|S| = |S_1| + |S_2| + |S_3| \geq (d^2 k + p) + (d k + s) + (k - s - i) = k' + (p - i)$.
Since $|S| \leq k'$ by assumption, we have $p = i$, which is possible only if $p = 0$. Hence, we have $X_1 = X$, and thus there exists $\K^0 \subseteq \C^0$ assignment such that $S_1 = \{ e^1_{x,a} : (x = a) \notin \K^0 \}$. We then have $|S_2| = d k + s$ with $0 \leq s \leq k$. Suppose that $s > 0$. It follows that $q \geq d (k - s) + s > d (k - s)$, and we must have $|S_3| > k - s$. But then $|S| = |S_1| + |S_2| + |S_3| > d^2 k + (d k + s) + (k - s) = k'$, contradiction. Thus, we must have $s = 0$.

We obtain that $|S_1| = d^2 k, |S_2| = d k$ and $|S_3| = k$. We then find $\K^1 \subseteq \C^1$ of size $dk$ such that $S_2 = \{ e^2_C : C \in \K^1 \}$, and $\K^2 \subseteq \C^2$ of size $k$ such that $S_3 = \{ e^3_C : C \in \K^2 \}$.
\qed
\end{proof}

\begin{theorem} \label{thm7} For every integer $d \geq 3$, \textsc{$d$-Flat Set Cover} has no compression of size $O(k^{d^2 - \epsilon})$, unless $\conp \subseteq \np / poly$.
\end{theorem}

\begin{proof} Given an instance $I$ of \textsc{$d$-Constraint-Cover}, we construct $I' = (H',k')$ instance of \textsc{$d$-Flat Set Cover} as above. By Lemma \ref{lem7}, $H'$ is $d$-square free, and thus $H'$ is $d$-flat. We show that the instances $I$ and $I'$ are equivalent. By Theorem \ref{thm6}, this implies that \textsc{$d$-Flat Set Cover} has no compression of size $O(k^{d^2 - \epsilon})$ unless $\conp \subseteq \np / poly$. 

Suppose that $I$ is a positive instance. We then find a solution consisting of an assignment $\K^0 \subseteq \C^0$, a set $\K^1 \subseteq \C^1$ of size $dk$, and a set $\K^2 \subseteq \C^2$ of size $k$.  We define the following sets: $S_1 = \{ e^1_{x,a} : x \in X \text{ and } (x = a) \notin \K^0 \}$, $S_2 = \{ e^2_C : C \in \K^1 \}$ and $S_3 = \{ e^3_C : C \in \K^2 \}$. We let $S = S_1 \cup S_2 \cup S_3$. Clearly, $|S| = k'$. We show that $S$ is a cover of $H'$. Consider a vertex $z \in V'$. If $z = v^1_x$, it is covered by the edge $e^1_{x,a} \in S_1$. If $z = v^2_{x,a}$, it is covered either by the edge $e^1_{x,a} \in S_1$ (if $(x = a) \notin \K^0$) or by an edge $e^2_C \in S_2$ (if $(x = a) \in \K^0$, by Point (i)). Suppose that $z = v^3_C$ for some $C \in \C^1$. If $C \nsubseteq \K^0$, then $C$ contains a 0-constraint $(x = a) \notin \K^0$; $z$ is then covered by the edge $e^1_{x,a} \in S_1$. If $C \subseteq \K^0$, we have $C \in \K^1$ by Point (iii); we then find $C' \in \K^2$ containing $C$ by Point (ii), and thus $z$ is covered by $e^3_{C'} \in S_3$.

Suppose that $I'$ is a positive instance. We then find $S$ cover of $H'$ of size at most $k'$. By Lemma \ref{lem8}, we obtain $\K^0 \subseteq \C^0$ assignment, $\K^1 \subseteq \C^1$ of size $dk$, $\K^2 \subseteq \C^2$ of size $k$, and a partition of $S$ into $S_1,S_2,S_3$ as stated. We show that the sets $\K^i$ form a solution for $I$. 

Let us show Point (i). Since $|\K^0| = d^2 k$ and $|\K^1| = d k$, it suffices to show that $\K^0 \subseteq \bigcup_{C \in \K^1} C$. Consider a constraint $(v = a) \in \K^0$. Since $e^1_{x,a} \notin S_1$, $v^2_{x,a}$ must be covered by an edge $e^2_C \in S_2$. We obtain $C \in \K^1$ containing $(v = a)$.

Let us show Point (ii). Since $|\K^1| = d k$ and $|\K^2| = k$, it suffices to show that $\K^1 \subseteq \bigcup_{C \in \K^2} C$. Consider a constraint $C \in \K^1$. By Point (i), we have $C \subseteq \K^0$. It follows that $v^3_C$ cannot be covered by an edge $e^1_{x,a} \in S_1$, and thus it must be covered by an edge $e^3_{C'} \in S_3$. We obtain $C' \in \K^2$ containing $C$.

Let us show Point (iii). Consider $C \in \C^1$ such that $C \subseteq \K^0$. It follows that $v^3_C$ cannot be covered by an edge of $S_1$. Thus, $v^3_C$ must be covered by an edge $e^3_{C'} \in S_3$. We then have $C' \in \K^2$ containing $C$. By Point (ii), we conclude that $C' \in \K^1$.
\qed
\end{proof}

\section{\label{sec5}Concluding remarks}

We have introduced the notion of \emph{$d$-flat hypergraph} and we have given an FPT algorithm and a kernelization for \textsc{Set Cover} on $d$-flat hypergraphs. We were also able to prove the tightness of our kernel size under standard complexity-theoretic assumptions.

An obvious question arising from our work concerns the tightness of our algorithms. Regarding Algorithm \ref{algo1}, it may be possible to rule out a $2^{o(k \log k)} ||H||^c$ running time using the proof technique of \cite{LMS18}. Regarding Algorithm \ref{algo2}, it seems unlikely that its running time could be improved to $f(d) n^{o(d)}$, although it is unclear how to rule this out with known techniques. Already improving it, for instance to $O(n^{d+1})$, seems difficult: to compute the values $l_{H'}(e)$, it seems necessary to examine every pair of edges of $H'$, which number can be $\Theta(n^d)$.

An interesting direction for future work would be to identify larger classes of hypergraphs for which the \textsc{Set Cover} problem is fixed-parameter tractable / kernelizable when parameterized by the solution size. A first goal would be to settle this question for the class of $d$-ladder-free hypergraphs; to the best of our knowledge, this is open already for $d = 3$ and is likely to require deeper structural insights on these classes. A more distant goal would be to seek a dichotomy for the parameterized complexity of \textsc{Set Cover} in $H$-free hypergraphs; this question seems ambitious but not unreasonable, given the progress in dichotomies for \textsc{Independent Set} \cite{GLMPPR23} or \textsc{Dominating Set} \cite{K90}.

\bibliography{Submission}
\bibliographystyle{ieeetr}

\section{Appendix}

In this section, we prove Theorem \ref{thm6} by giving a parameterized reduction from \textsc{$s$-Dimensional Matching}. We first define the problem. Let $V$ be a ground set and two tuples $t,t' \in V^s$; we say that $t,t'$ are \emph{orthogonal} iff $t[i] \neq t'[i]$ for each $i \in [s]$. A set $M \subseteq V^s$ is a \emph{matching} iff the tuples in $M$ are pairwise orthogonal. The \textsc{$s$-Dimensional Matching} problem takes a ground set $V$, a set $S \subseteq V^s$ and an integer $l$, and asks whether $S$ contains a matching of size $l$. We have the following compression lower bound for the problem due to \cite{DM12}.

\begin{theorem} \label{thm8} For every integer $s \geq 3$, \textsc{$s$-Dimensional Matching} has no compression of size $O(|V|^{s-\epsilon})$, unless $\conp \subseteq \np / poly$.
\end{theorem}

The proof of Theorem \ref{thm6} proceeds by contradiction. We suppose that there exists an integer $d \geq 3$ and a real $\epsilon > 0$ such that \textsc{$d$-Constraint-Cover} has a compression $\A$ of size $O(k^{d^2 - \epsilon})$. We choose an integer $c$ such that $d^2 - (c+1) \epsilon < 0$, and we let $s = c d^2$. We give a polynomial-time reduction that maps an instance $I = (V,S,l)$ of \textsc{$s$-Dimensional Matching} to an instance $I'$ of \textsc{$d$-Constraint-Cover}; composing this reduction with algorithm $\A$ will then yield a compression of instance $I$ to an instance of size $|V|^{s-\epsilon'}$, thus implying $\conp \subseteq \np / poly$ by Theorem \ref{thm8}.

Consider an instance $I = (V,S,l)$ of \textsc{$s$-Dimensional Matching}, where $S \subseteq V^s$. We construct an instance $I'$ of \textsc{$d$-Constraint-Cover} as follows. Let $n = |V|$ and $m = n^c$. By padding, we may assume that $n \equiv 1~(d)$ and $m \equiv 1~(d)$. In the following, we will consider a fixed bijection $\Phi : [0;m-1] \rightarrow V^c$. We define the parameter $k = d l m$. The set $X$ has size $d^3 l m$ and is partitioned in sets $X_p$ ($p \in [l]$), each of size $d^3 m$. A set $X_p$ is partitioned in sets $X_{p,q,r}$ ($q,r \in [d]$), each of size $d m$. A set $X_{p,q,r}$ contains the vertices $x_{p,q,r,i,j}$ ($i \in [0;m-1], j \in [d]$).\\

We now describe the construction of the sets of constraints $\C^i$ involved in instance $I'$. We let $\C^0$ be the full set of 0-constraints over $X$. To construct $\C^1$ and $\C^2$, we start with $\C^1 = \C^2 = \emptyset$, and we add the following constraints.

First type: we add a set of \emph{incompatibility} constraints. Fix $q,r \in [d]$. For each $p \in [l]$, $i,i' \in [0;m-1]$ distinct, $j,j' \in [d]$, we add to $\C^1$ the constraint $C^1_{p,q,r,i,i',j,j'} = \{ x_{p,q,r,i,j} = 1, x_{p,q,r,i',j'} = 1 \}$. For each $p,p' \in [l]$ distinct, $i,i' \in [0;m-1]$ such that $\Phi(i), \Phi(i')$ are not orthogonal, $j,j' \in [d]$, we add to $\C^1$ the constraint $C^2_{p,p',q,r,i,i',j,j'} = \{ x_{p,q,r,i,j} = 1, x_{p',q,r,i',j'} = 1 \}$.

Second type: we add a set of \emph{local} constraints inside each block $X_{p,q,r}$. Fix $p \in [l]$ and $q,r \in [d]$. For each $i \in [0;m-1]$, we add to $\C^1$ the constraint $C^3_{p,q,r,i} = \{ x_{p,q,r,i,j} = 0 : 1 \leq j \leq d \}$. For each $i \in [0;m-1]$, we add to $\C^2$ the constraint $C^4_{p,q,r,i} = \{ C^3_{p,q,r,(i+z) \mod m} : 0 \leq z < d \}$.

Third type: we add a set of \emph{global} constraints. Fix $p \in [l]$, $q \in [d]$, $t \in V^{c d}$ and $j \in [d]$. Consider the factorization $t = t_1 \ldots t_d$ where $t_1,\ldots,t_d \in V^c$. For each $r \in [d]$, let $i_r = \Phi^{-1}(t_r)$. We add to $\C^1$ the constraint $C^5_{p,q,t,j} = \{ x_{p,q,r,i_r,j} = 1 : r \in [d] \}$. Fix $p \in [l]$, $t \in V^s$ and $j \in [d]$. Consider the factorization $t = t_1 \ldots t_d$ where $t_1,\ldots,t_d \in V^{c d}$. We add to $\C^2$ the constraint $C^6_{p,t,j} = \{ C^5_{p,q,t_q,j} : q \in [d] \}$.\\

The intuition behind the reduction is as follows. We can view the blocks $X_{p,q,r}$ as arranged on a grid, where each column is indexed by $(q,r) \in [d] \times [d]$ and each row is indexed by $p \in [l]$. The row of index $p$ encodes the choice of a tuple $t_p \in V^s$, such that $M = \{ t_1,\ldots,t_l \}$ is a matching included in $S$. A given tuple $t_p$ is decomposed in factors $t_{p,q,r} \in V^c$, for $q,r \in [d]$. Each factor $t_{p,q,r}$ is encoded by the choice of an index $i \in [0;m-1]$ inside block $X_{p,q,r}$, such that $\Phi(i) = t_{p,q,r}$.

Consider a solution for $I'$, consisting of assignment $\K^0 \subseteq \C^0$, and of sets $\K^1 \subseteq \C^1$ and $\K^2 \subseteq \C^2$. We need to ensure that for each block $X_{p,q,r}$, we choose a $t \in V^c$ such that if $i = \Phi(t)$ then $\K^0$ contains $(x_{p,q,r,i,j} = 1)$ ($j \in [d]$) and $(x_{p,q,r,i',j} = 0)$ ($i' \neq i, j \in [d]$). We also need to ensure that for two blocks in the same column, i.e. $X_{p,q,r}$ and $X_{p',q,r}$, the tuples chosen for these two blocks are orthogonal. This is the role of the constraints of the first type. Note that these constraints cannot be satisfied, as the satisfied constraints have arity $d > 2$.

Let us fix $p \in [l]$ and suppose that we have chosen a tuple $t_{p,q,r} \in V^c$ for each $q,r \in [d]$. For each $q \in [d]$, let $t_{p,q} = t_{p,q,1} \ldots t_{p,q,d}$, and let $t_p = t_{p,1} \ldots t_{p,d}$. For $i \in \{0,1\}$, let $X^i_p$ be the set of vertices $x \in X_p$ such that $\K^0$ contains $(x = i)$. To cover the vertices in $X_p$, we proceed as in Figure \ref{fig1}. We cover the vertices in $X^0_p$ using the constraints $C^3_{p,q,r,i}$ and $C^4_{p,q,r,i}$. We cover the vertices in $X^1_p$ using the constraints $C^5$ arising from the tuples $t_{p,q}$, and the constraints $C^6$ arising from the tuple $t_p$. In particular, each such constraint will ensure that $t_p$ is an element of $S$.\\

\begin{figure}[hbt]
\center
\includegraphics[scale=0.3]{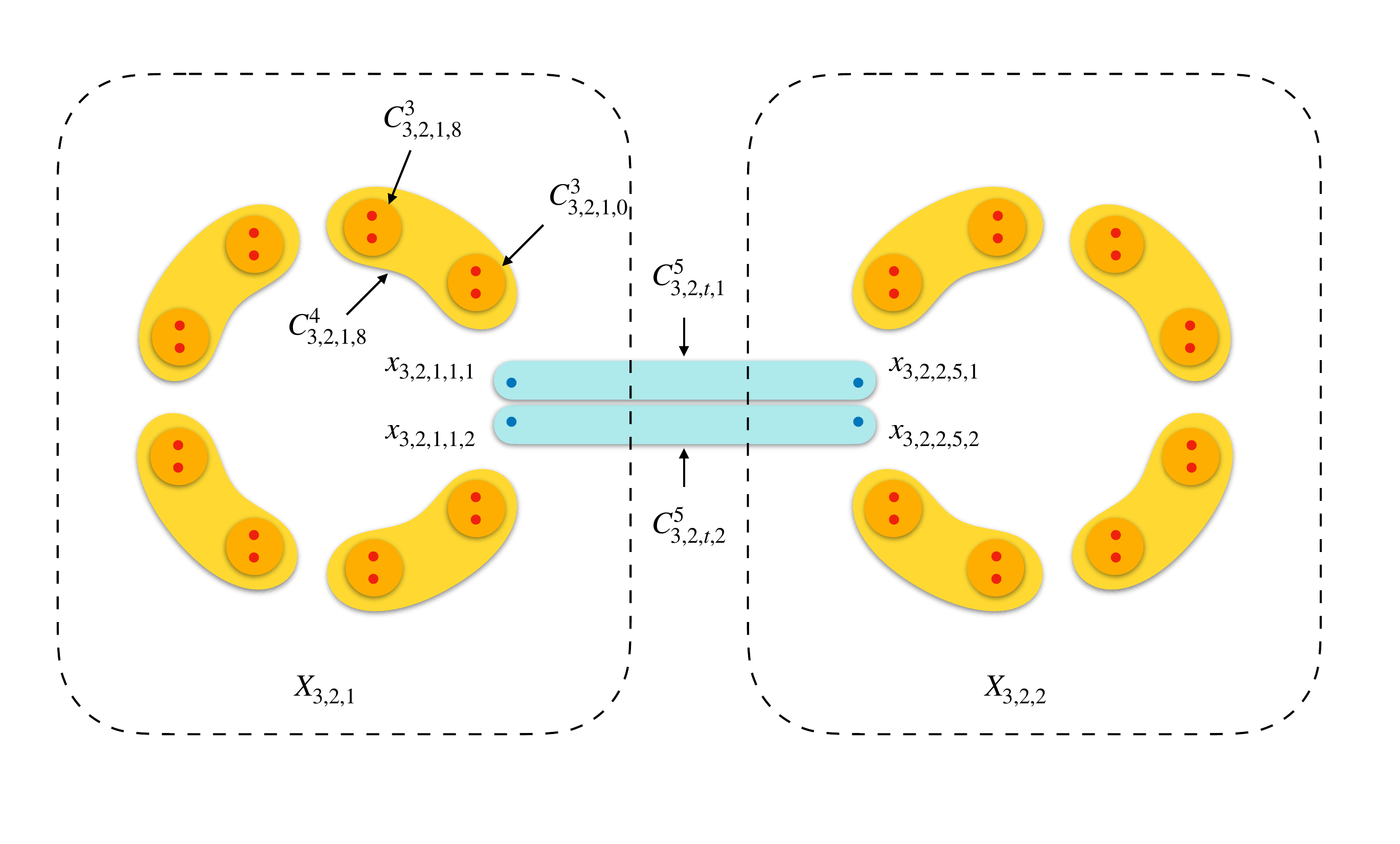}
\caption{An illustration of the construction for blocks $X_{3,2,1}$ and $X_{3,2,2}$. We assume that $d = 2$, $c = 1$, $|V| = 9$ and $t = [1,5]$. In a solution, the red vertices are labelled 0, the blue vertices are labelled 1, the orange and blue constraints are in $\K^1$, and the yellow constraints are in $\K^2$.}
\label{fig1}
\end{figure}

The following two lemmas prove the equivalence of the instances $I$ and $I'$. 

\begin{lemma} \label{lem9} If $I$ is a positive instance, then $I'$ is a positive instance.
\end{lemma}

\begin{proof} Suppose that we have $M = \{ t_1,\ldots,t_l \}$ matching included in $S$. For each $p \in [l]$, suppose that $t_p = t_{p,1} \ldots t_{p,d}$ with $t_{p,q} \in V^{c d}$ for each $q \in [d]$. For each $p \in [l], q \in [d]$, suppose that $t_{p,q} = t_{p,q,1} \ldots t_{p,q,d}$ with $t_{p,q,r} \in V^c$ for each $r \in [d]$. 

We construct the following solution of $I'$. We define the assignment $\K^0 \subseteq \C^0$ such that: for $p \in [l]$, $q,r \in [d]$, $i \in [0;m-1]$ and $j \in [d]$, $\K^0$ contains $(x_{p,q,r,i,j} = 1)$ iff $\Phi(i) = t_{p,q,r}$. We construct $\K^1$ and $\K^2$ as follows.
\begin{enumerate}
\item Local constraints. Fix $p \in [l]$ and $q,r \in [d]$. Let $i = \Phi^{-1}(t_{p,q,r})$. For each $i' \in [0;m-1]$ different from $i$, we add $C^3_{p,q,r,i'}$ to $\K^1$. For each $i' \in [0;m-1]$ such that ($i' > i$ and $i' \equiv i+1~(d)$) or ($i' < i$ and $i' \equiv i~(d)$), we add $C^4_{p,q,r,i'}$ to $\K^2$.
\item Global constraints. Fix $p \in [l]$ and $j \in [d]$. For each $q \in [d]$, we add $C^5_{p,q,t_{p,q},j}$ to $\K^1$, and we add $C^6_{p,t_p,j}$ to $\K^2$.
\end{enumerate}
We verify that this is indeed a solution of $I'$. First, we have $|\K^1| = d^2 l (m - 1) + d^2 l = d^2 l m = d k$. Second, we have $|\K^2| = d^2 l (m - 1) / d + d l = d l (m - 1) + d l = d l m = k$. Third, it is easily checked that Points (i)-(ii)-(iii) hold. \qed
\end{proof}

\begin{lemma} \label{lem10} If $I'$ is a positive instance, then $I$ is a positive instance.
\end{lemma}

\begin{proof} Suppose that we have a solution of $I'$, consisting of assignment $\K^0 \subseteq \C^0$, of $\K^1 \subseteq \C^1$ and of $\K^2 \subseteq \C^2$. We then have $|\K^1| = d k = d^2 l m$ and $|\K^2| = k = d l m$. For $p \in [l]$ and $q,r \in [d]$, we let $Z_{p,q,r}$ be the set of pairs $(i,j)$ ($i \in [0;m-1]$, $j \in [d]$) such that $\K^0$ contains $(x_{p,q,r,i,j} = 1)$.\\

\noindent \textit{Claim 1.} For each $p \in [l]$, $q,r \in [d]$, we have some $i \in [0;m-1]$ such that $Z_{p,q,r} = \{i\} \times [d]$. 

\noindent \textit{Proof.} We first show that we cannot have $Z_{p,q,r} = \emptyset$. Suppose the contrary. Let $\K'$ be the set of constraints $C^3_{p,q,r,i}$ ($i \in [0;m-1]$). By assumption, the constraints $(x_{p,q,r,i,j} = 0)$ are in $\K^0$, and by Point (i) we must have $\K' \subseteq \K^1$. By Point (ii), it follows that there is a set of indices $I \subseteq [0;m-1]$ such that $\K^2$ contains constraints $C^4_{p,q,r,i}$ ($i \in I$). We then have $\K' = \cup_{i \in I} C^4_{p,q,r,i}$, and since this is a disjoint union, we obtain $|\K'| = d |I| < m$, contradiction.

We now show that we cannot have two tuples $(i,j), (i',j') \in Z_{p,q,r}$ with $i \neq i'$. Suppose the contrary. It follows that $\K^0$ contains $(x_{p,q,r,i,j} = 1)$ and $(x_{p,q,r,i',j'} = 1)$, hence we have $C^1_{p,q,r,i,i',j,j'}$ satisfied by $\K^0$. This contradicts Point (iii) since the constraints of $\C^1$ satisfied by $\K^0$ have arity $d > 2$.

From the previous points, it follows that there is $i \in [0;m-1]$ such that $Z_{p,q,r} \subseteq \{i\} \times [d]$. We show that the inclusion is an equality. Suppose the contrary, we then find $j,j' \in [d]$ such that $\K^0$ contains $(x_{p,q,r,i,j} = 0)$ and $(x_{p,q,r,i,j'} = 1)$. By Point (i), the constraint $(x_{p,q,r,i,j} = 0)$ can only be covered by a constraint of the form $C = C^3_{p,q,r,i}$. Since $\K^0$ does not contain $(x_{p,q,r,i,j'} = 0)$, we have $C \nsubseteq \K^0$, contradiction. \qed
~\\

For each $p \in [l]$ and $q,r \in [d]$, the previous observation then yields an integer $i \in [0;m-1]$ such that $Z_{p,q,r} = \{i\} \times [d]$; we then define the tuple $t_{p,q,r} \in V^c$ such that $\Phi(i) = t_{p,q,r}$. For $p \in [l]$ and $q \in [d]$, we define  $t_{p,q} \in V^{c d}$ such that $t_{p,q} = t_{p,q,1} t_{p,q,2} \ldots t_{p,q,d}$. For each $p \in [l]$, we define $t_p \in V^{s}$ such that $t_p = t_{p,1} t_{p,2} \ldots t_{p,d}$. \\

\noindent \textit{Claim 2.} $\K^1$ contains: (a) the constraint $C^3_{p,q,r,i}$ for each $p \in [l]$, $q,r \in [d]$, $i \neq \Phi^{-1}(t_{p,q,r})$, (b) the constraint $C^5_{p,q,t_{p,q},j}$ for each $p \in [l]$, $q \in [d]$, $j \in [d]$, (c) no other constraint.

\noindent \textit{Proof.} Point (a). Consider $p \in [l], q,r \in [d], c \in [p]$ and $i \neq \Phi^{-1}(t_{p,q,r})$. By definition, $\K^0$ contains $(x_{p,q,r,i,j} = 0)$ for each $j \in [d]$. By Point (i), the only possibility to cover these $d$ constraints is to have $C^3_{p,q,r,i}$ in $\K^1$.

Point (b). Let $\K'$ denote the set of constraints in $\K^1$ not considered in (a). We have $|\K'| = |\K^1| - d^2 l (m - 1) = d^2 l$. For each $p \in [l]$ and $q,j \in [d]$, let $\K'_{p,q,j}$ denote the set of constraints in $\K'$ of the form $C ^5_{p,q,t,j}$. Fix $p \in [l]$ and $q,j \in [d]$, and for $r \in [d]$ let $i_r = \Phi^{-1}(t_{p,q,r})$. By Claim 1, $\K^0$ contains $(x_{p,q,r,i_r,j} = 1)$ for each $r \in [d]$. Each such constraint must be covered by a constraint of $\K'_{p,q,j}$. Since $|\K'| = d^2 l$, it follows that $|\K'_{p,q,j}| = 1$, and thus we find $t'_{p,q,j} \in V^s$ such that $\K'_{p,q,j} = \{C^5_{p,q,t'_{p,q,j},j}\}$. The definition of this constraint ensures that $t'_{p,q,j} = t_{p,q}$ for each $j \in [d]$.

Point (c) follows from the fact that $|\K^1| = d^2 l m$. \qed
~\\

Let $M = \{ t_1,\ldots,t_l \}$. Since each $t_p$ is in $V^s$ by construction, we have $M \subseteq V^s$. It remains to show that $M$ is the desired solution of instance $I$.\\

\noindent \textit{Claim 3.} $M$ is a matching included in $S$.

\noindent \textit{Proof.} We first show that: for each $p \in [l], j \in [d]$, if $\K^2$ contains $C^6_{p,t',j}$ then $t' = t_p$. Suppose that $\K^2$ contains $C^5_{p,t',j}$ with $t' \in V^s$, and suppose that $t' = t'_1 \ldots t'_d$ where $t'_1, \ldots, t'_d \in V^{c d}$. By Point (ii), we have $C^5_{p,q,t'_q,j} \in \K^1$ for each $q \in [d]$. By Claim 2, we obtain that $t'_q = t_{p,q}$ for each $q \in [d]$. It follows that $t' = t_p$.

We show that $M$ is a matching included in $S$. Fix $p \in [l]$. Since the constraints $C^5_{p,q,t_{p,q},j}$ are in $\K^1$ (by Claim 2), and by the previous reasoning, we obtain that $\K^2$ contains $C^6_{p,t_p,j}$. It follows that $t_p \in S$ for each $p \in [l]$, and thus $M \subseteq S$. We now prove that the elements of $M$ are pairwise orthogonal. Suppose by contradiction that there exists $p,p' \in [l]$ distinct such that $t_{p}, t_{p'}$ are not orthogonal. We then find $q,r \in [d]$ such that $t_{p,q,r}$ and $t_{p',q,r}$ are not orthogonal. Let $i = \Phi(t_{p,q,r})$ and $i' = \Phi(t_{p',q,r})$. We obtain that $\K^0$ contains $(x_{p,q,r,i,1} = 1)$ and $(x_{p',q,r,i',1} = 1)$. This implies that the constraint $C^2_{p,p',q,r,i,i',1,1}$ is satisfied by $\K^0$. This contradicts Point (iii) since the constraints of $\C^1$ satisfied by $\K^0$ have arity $d > 2$. \qed
\end{proof}

We are now ready to finish the proof of Theorem \ref{thm6}. We obtain a compression of \textsc{$s$-Dimensional Matching} as follows. Given the instance $I = (V,S,l)$, we construct the instance $I'$ of \textsc{$d$-Constraint-Cover} as above; the instances $I$ and $I'$ are equivalent by Lemmas \ref{lem9} and \ref{lem10}. By the assumption, we may compress $I'$ to an instance $I''$ of size $O(k^{d^2-\epsilon})$. Since $k = d l m \leq d n^{c+1}$, it follows that $I''$ has size $O(n^{(c+1)(d^2 - \epsilon)}) = O(n^{s + d^2 - (c+1) \epsilon})$. Since $d^2 - (c+1) \epsilon < 0$ by choice of $c$, we have compressed instance $I$ to an instance $I''$ of size $O(|V|^{s - \epsilon'})$. By Theorem \ref{thm8}, this implies $\conp \subseteq \np / poly$.

\end{document}